\documentclass[10pt,pra,twocolumn]{revtex4}
\usepackage{hyperref}
\usepackage{latexsym,amsmath,amssymb,amsfonts,graphicx,color,amsthm}
\usepackage{enumerate}
\usepackage{verbatim}
\usepackage{subfig}

\newcommand{\cA}{\mathcal{A}}
\newcommand{\cB}{\mathcal{B}}
\newcommand{\cC}{\mathcal{C}}

\newcommand{\cI}{\mathcal{I}}

\newcommand{\cS}{\mathcal{S}}

\newcommand{\cU}{\mathcal{U}}

\newcommand{\cX}{\mathcal{X}}
\newcommand{\cY}{\mathcal{Y}}
\newcommand{\cZ}{\mathcal{Z}}

\newcommand{\tr}{\text{tr}}

\newtheorem{theorem}{Theorem}
\newtheorem*{theorem*}{Theorem}
\newtheorem{lemma}[theorem]{Lemma}
\newtheorem{property}{Property}
\newtheorem*{property*}{Property}

\newtheorem{definition}[theorem]{Definition}

\begin{document}

\title{Unextendible mutually unbiased bases in prime-squared dimensions}
\author{Vishakh Hegde}
\affiliation{Department of Physics, IIT Madras, Chennai - 600036, India}
\author{Prabha Mandayam}
\affiliation{Department of Physics, IIT Madras, Chennai - 600036, India}

\date{\today}
\begin{abstract}
A set of mutually unbiased bases (MUBs) is said to be {\it unextendible} if there does not exist another basis that is
unbiased with respect to the given set. Here, we prove the existence of smaller sets of MUBs in prime-squared dimensions ($d=p^{2}$) that cannot be extended to a complete set using the generalized Pauli operators. We further observe an interesting connection between the existence of unextendible sets and the tightness of entropic uncertainty relations (EURs) in these dimensions. In particular, we show that our construction of unextendible sets of MUBs naturally leads to sets of $p+1$ MUBs that saturate both a Shannon ($H_{1}$) and a collision ($H_{2}$) entropic lower bound. Such an identification of smaller sets of MUBs satisfying tight EURs is crucial for cryptographic applications as well as constructing optimal entanglement witnesses for higher dimensional systems.
\end{abstract}

\maketitle

Two orthonormal bases $\cA = \{|a_{i}\rangle, i=1,\ldots, d\}$ and $\cB = \{|b_{j}\rangle, j=1, \ldots, d\}$ of a $d$-dimensional Hilbert space $\mathbb{C}^{d}$ are said to be {\bf mutually unbiased} if for all basis vectors $|a_{i}\rangle \in \cA$ and $|b_{j}\rangle\in \cB$,
\begin{equation}
\vert\langle a_{i}|b_{j}\rangle\vert = \frac{1}{\sqrt{d}}, \forall i, j = 1,\ldots, d.
\end{equation}
In physical terms, if a system is prepared in an eigenstate of basis $\cA$ and measured in basis $\cB$, all outcomes
are equally probable. A set of orthonormal bases $\{\cB_{1}, \cB_{2}, \ldots, \cB_{m}\}$ in $\mathbb{C}^{d}$ is called a set mutually unbiased bases (MUBs) if every pair of bases in the set is mutually unbiased. MUBs form a minimal and optimal set of orthogonal measurements for quantum state tomography~\cite{Ivanovic81, WF89}. Such bases play an important role in our understanding of complementarity in quantum mechanics~\cite{mubreview_2010}
and are central to quantum information tasks such as entanglement detection~\cite{spengler_entanglement}, information
locking~\cite{locking04}, and quantum cryptography~\cite{bb84, KWW12}.

MUBs correspond to measurement bases that are most `incompatible', as quantified by uncertainty relations~\cite{WWsurvey} and other incompatibility
measures~\cite{SB_PM13, PM_MDS14}, and, the security of quantum cryptographic tasks relies on this property of MUBs. In particular, protocols based on higher-dimensional quantum systems with larger numbers of unbiased basis sets can have certain advantages over those based on qubits~\cite{dlevel_qkd,MW11}. However, beyond the case of two measurements, being mutually unbiased is a necessary but not sufficient condition for satisfying a strong entropic lower bound~\cite{BW07}. It is therefore important for cryptographic applications to identify sets of MUBs in higher-dimensional systems that satisfy strong uncertainty relations.

The maximum number of MUBs that can exist in a $d$-dimensional Hilbert space is $d+1$ and explicit constructions of such complete sets are known when $d$ is a prime power~\cite{WF89, BBRV02, LBZ02}. However, in non-prime-power dimensions, the question of whether a complete set of MUBs exists remains unresolved. Related to the question of finding complete sets of MUBs is the important concept of {\it unextendible} sets of MUBs. A set of MUBs $\{\cB_{1}, \cB_{2}, \ldots,\cB_{m}\}$ in $\mathbb{C}^{d}$ is said to be {\bf unextendible} if there does not exist another basis in $\mathbb{C}^{d}$ that is unbiased with respect to all the bases $\cB_{j}, j=1,\ldots,m$. Examples of such unextendible sets are known in the literature~\cite{Grassl04, Brierley09, Jaming09, McNulty12, BSTW05}.

More recently, a systematic construction of such smaller sets that are unextendible to a complete set was obtained for two- and three-qubit systems~\cite{unext_MBGW}. In the case of two-qubit systems, an interesting connection was noted between unextendible sets of Pauli classes and state-independent proofs of the Kochen-Specker Theorem. It was also shown that the tightness of the an entropic uncertainty relation for any set of three MUBs in $d=4$ follows as an important consequence of the existence of weakly unextendible sets of MUBs~\cite{unext_MBGW}. The existence of similar unextendible sets was conjectured for $d=2^{n} (n>3)$. This conjecture has now been further improved upon~\cite{thas14} using a correspondence between unextendible sets of MUBs and maximal partial spreads of the polar space formed by the $n$-qubit Pauli operators~\cite{thas09}.

Here, we provide a construction of weakly unextendible sets of MUBs in prime-squared dimensions $d=p^{2}$, where $p$ is prime. Each MUB is realized as the common eigenbasis of a maximal commuting class of tensor products of the generalized Pauli operators. Our construction also brings to light an interesting connection between the existence of unextendible sets and the tightness of entropic lower bounds in these dimensions. In particular, we identify sets of $p+1$ MUBs that saturate both a Shannon and a collision EUR in $d=p^{2}$. This has important consequences for both cryptographic applications and for constructing entanglement witnesses in higher dimensional systems.

The rest of the paper is organized as follows. We begin with a brief review of the standard construction of MUBs in Sec.~\ref{sec:prelims} and formally define the notion of unextendibility. We state our main result on the construction of unextendible sets of MUBs in Sec.~\ref{sec:unext_psq} and provide proofs in the appendix (\ref{sec:appendix2}). Finally, in Sec.~\ref{sec:tightEUR},  we note the connection between the existence of unextendible MUBs and the tightness of EURs in prime-squared dimensions.  

\section{Preliminaries}\label{sec:prelims}

Our construction of unextendible MUBs is based on the well known connection between mutually unbiased bases and mutually disjoint maximal commuting operator classes~\cite{BBRV02}. Consider a set $\cS$ of $d^2$ mutually orthogonal unitary operators in a $d$-dimensional Hilbert space $\mathbb{C}^d$. Such a set constitutes a basis for $\mathbb{M}_d(\mathbb{C})$, the space of $d\times d$ complex matrices. Since at most $d$ such operators can mutually commute, we may consider a partitioning of the operator basis $\cS$ into mutually disjoint maximal commuting classes as follows.

\begin{definition}[Mutually Disjoint Maximal Commuting Classes]
 A set of subsets ${\cC_1, \cC_2, \ldots, \cC_L | \cC_j  \subset \cS\setminus\{\cI\}}$ of size $|\cC_j|=d-1$ constitutes a (partial) partitioning of $\cS \setminus\{\cI\}$ into mutually disjoint maximal commuting classes if the subsets $\cC_j$ are such that
 \begin{itemize}
  \item The elements of $\cC_j$ commute $\forall ~1 \leq j \leq L$
  \item $\cC_j \cap \cC_k = \phi ~ \forall ~ j\neq k$
 \end{itemize}
\end{definition}
The existence of such a partitioning of the operator basis $\cS$ is directly related to the existence of mutually unbiased bases. We formally state this result in the following Lemma, and refer to~\cite{BBRV02} for the proof.
\begin{lemma}\label{lem:MUB_MCC}
Let $\cS$ be any unitary operator basis for $\mathbb{M}_d(\mathbb{C})$. There exist a set of $L$ mutually unbiased bases in $\mathbb{C}^{d}$ iff the $\cS\setminus\{\cI\}$ can be partitioned into $L$ mutually disjoint maximal commuting classes. Furthermore, the MUBs are simply realized as the common eigenvectors of the different maximal commuting operator classes.
\end{lemma}
Since the maximum number of such classes that can be formed in $d$-dimensions is $d+1$, it follows that the number of MUBs in $\mathbb{C}^{d}$ is at most $d+1$. This bound is saturated for prime power dimensions~\cite{WF89}.

A simple example of such a unitary operator basis $\cS$ is the one comprising of products of the generalized Pauli operators acting on $\mathbb{C}^{d}$, which are defined as:
\begin{align}
\cX_{d} \vert j \rangle &= \vert (j+1) \; {\rm mod} \; d \rangle \nonumber \\
\cZ_{d} \vert j \rangle &= \omega^{j} \vert j \rangle, \label{eq:gen_Pauli}
\end{align}
where $\omega = e^{\frac{2\pi i }{d}}$. 
We will in fact make use of the unitary basis generated by the generalized Paulis in prime-dimensions for our construction of unextendible sets.


\subsection{Unextendible sets of MUBs and Maximal Commuting Operator Classes}

We now proceed to formally define the notion of unextendibility of MUBs, and the related notion of unextendible sets of operator classes.

\begin{definition}[Unextendible Sets of MUBs]
A set of MUBs $\{\cB_1, \cB_2, \ldots, \cB_{L}\}$ in $\mathbb{C}^{d}$ is said to be unextendible if there does not exist another basis in $\mathbb{C}^{d}$ which is unbiased with respect to all the bases in the set.
\end{definition}
For example, in dimension $d=6$, the eigenbases of $\cX_{6}, \cZ_{6}$ and $\cX_{6}\cZ_{6}$ were shown to be an unextendible set of MUBs~\cite{Grassl04}. This has the important consequence that the eigenbases of Weyl-Hiesenberg generators will not lead to a complete set of $7$ MUBs in $d=6$. In fact, several distinct families of unextendible triplets of MUBs have been constructed in $d=6$~\cite{Brierley09, Jaming09, McNulty12}. Moving away from six dimensions, the set of three MUBs obtained in $d=4$ using Mutually Orthogonal Latin Squares (MOLS)~\cite{WB05} is an example of an unextendible set of MUBs in prime-power dimensions~\cite{BSTW05}.

If there does not exist any vector $v \in \mathbb{C}^d$ that is unbiased with respect to the MUBs  $\{\cB_1, \cB_2, \ldots, \cB_{L}\}$ in $\mathbb{C}^d$, then the set of MUBs are said to be {\it strongly unextendible}. It has been shown that the eigenbases of $\cX_6$, $\cZ_6$ and $\cX_6\cZ_6$ are in fact strongly unextendible~\cite{Grassl04}.

A possible approach to constructing such unextendible sets of MUBs is to start with maximal commuting classes of operators which are unextendible in the following sense.
\begin{definition}[Unextendible Sets of Operator Classes]
A set of mutually disjoint maximal commuting classes $\cC_1, \cC_2, \ldots, \cC_{L}$ of operators drawn from a unitary basis $\cS$ is said to be unextendible if no other maximal class can be formed out of the remaining operators in $\cS \setminus (\{I\} \cup \bigcup_{i=1}^{L}\cC_i)$
\end{definition}
The eigenbases $\{\cB_1, \cB_2, \ldots, \cB_{L}\}$ of the operator classes $\{\cC_{1}, \cC_{2}, \ldots, \cC_{L} \subset \cS\}$ form a set of $L$ {\it weakly unextendible MUBs} in the following sense : There does not exist another basis unbiased with respect to $\{\cB_{1}, \cB_{2}, \ldots, \cB_{L}\}$ that can be obtained as the common eigenbasis of a maximal commuting class of operators in $\cS$.

For example, consider the space $\mathbb{C}^4 = \mathbb{C}^2 \otimes \mathbb{C}^2$. The Pauli operators $\cX_{2}$, $\cZ_{2}$, $\cY_{2} = i\cX_{2}\cZ_{2}$ and their tensor products give rise to a set of $16$ orthogonal two-qubit unitaries, including the identity operator $\cI_4$. It is known these can be partitioned into a set of five mutually disjoint maximal commuting classes~\cite{BBRV02, LBZ02} as for example,
\begin{eqnarray}
\cS_{1} &=& \{\cZ_{2}\otimes \cI_{2}, \cI_{2}\otimes \cZ_{2}, \cZ_{2}\otimes \cZ_{2} \}\nonumber \\
\cS_{2} &=& \{\cX_{2}\otimes \cI_{2}, \cI_{2}\otimes \cX_{2}, \cX_{2}\otimes \cX_{2}\} \nonumber \\
\cS_{3} &=& \{ \cX_{2} \otimes \cZ_{2}, \cZ_{2}\otimes \cY_{2}, \cY_{2}\otimes \cX_{2}\} \nonumber \\
\cS_{4} &=& \{\cY_{2}\otimes \cI_{2}, \cI_{2}\otimes \cY_{2}, \cY_{2}\otimes \cY_{2}\} \nonumber \\
\cS_{5} &=& \{\cY_{2}\otimes \cZ_{2}, \cZ_{2}\otimes \cX_{2}, \cX_{2} \otimes \cY_{2}\},
\end{eqnarray}
thus giving rise to a set of five MUBS in $d=4$.

Suppose we consider the following set of operator classes instead:
\begin{align}
 \cC_1 &= \{\cY_{2} \otimes \cY_{2}, \cI_{2} \otimes \cY_{2}, \cY_{2} \otimes \cI_{2}\} \nonumber \\
 \cC_2 &= \{\cY_{2} \otimes \cZ_{2}, \cZ_{2}\otimes \cX_{2}, \cX_{2}\otimes \cY_{2} \} \nonumber \\
 \cC_3 &= \{\cX_{2} \otimes \cI_{2} , \cI_{2} \otimes \cZ_{2}, \cX_{2}\otimes \cZ_{2}\}.  \label{eq:d4_example}
\end{align}
The above partitioning makes use of just $9$ of the $15$ possible two-qubit Pauli operators. It is easy to see that this partitioning gives rise to an {\it unextendible} set of classes, since is not possible to form another maximal commuting class from the remaining six operators:
\begin{equation*}
 \{\cI_{2} \otimes \cX_{2}, \cX_{2}\otimes \cX_{2}, \cY_{2}\otimes \cX_{2}\,
\cZ_{2} \otimes \cI_{2}, \cZ_{2}\otimes \cY_{2}, \cZ_{2}\otimes \cZ_{2}\}.
\end{equation*}
The common eigenbases of $\cC_{1}, \cC_{2}, \cC_{3}$ constitute a set of three weakly unextendible MUBs, as defined above. A systematic construction of such unextendible sets of classes in $d=2^{2}, 2^{3}$ was obtained recently~\cite{unext_MBGW}, and the corresponding MUBs were shown to be strongly unextendible.

\section{Unextendible sets of classes in prime-squared dimensions}\label{sec:unext_psq}

Here we examine whether it is possible to obtain a general construction of unextendible operator classes leading to unextendible MUBs in prime-power dimensions. The unitary basis of interest here is the one generated by tensor products of the generalized Paulis $\cX_{p}$ and $\cZ_{p}$ acting on a quantum systems of prime dimensions $p$ as specified in Eq.~\eqref{eq:gen_Pauli}:
\begin{align*}
\cX_{p} \vert j \rangle &= \vert (j+1) \; {\rm mod} \; p \rangle, \nonumber \\
\cZ_{p} \vert j \rangle &= \omega^{j} \vert j \rangle, \; \omega = e^{\frac{2\pi i }{p}} .
\end{align*}
In particular, restricting our attention to prime-squared dimensions ($d=p^{2}$, $p$ is prime), we consider the unitary operator basis $\cU^{(p^{2})}$  comprising operators of the form
\[ U = (\cX_{p})^{m}(\cZ_{p})^{n}\otimes(\cX_{p})^{k}(\cZ_{p})^{l}, \quad m,n,k,l \in \mathbb{F}_{p}, \]
where $\mathbb{F}_{p}$ is the prime field of order $p$. Every $U \in \cU^{(p^{2})}$ satisfies $(U)^{p} = \cI_{p^{2}}$, where $\cI_{p^{2}}$ denotes the identity operator in the $p^{2}$-dimensional space. We show via explicit construction that using the operators in $\cU^{(p^2)}\setminus\{\cI_{p^{2}}\}$ it is indeed possible to construct unextendible sets of operator classes of cardinalities 
\[ N(p) = p^{2}-p+1, p^{2}-p+2 \] 
in prime-squared dimensions.


\subsection{Structure of operator classes in $d=p^{2}$}

Our construction primarily relies on the properties of maximal commuting classes constructed out of operators in $\cU^{(p^{2})}\setminus\{\cI_{p^{2}}\}$. We first list some of these properties and prove a few simple consequences of these properties, which are useful for our construction. 

\begin{property}\label{prop:one}
Every maximal commuting class $\cC$ of operators in $\cU^{(p^2)}$ is generated by a set of $p+1$ independent operators $\{U_{1},U_{2}, U_{1}U_{2}, U_{1}^{2}U_{2}, \ldots, U_{1}^{p-1}U_{2}\} \in \cU^{(p^2)}$. 
\end{property}
Note that $U_{1}$ and $U_{2}$ are said to be {\it independent} if there do not exist $k,l \in \mathbb{F}_{p}$ such that $U_{1}^{k} = U_{2}^{l}$. To verify the above property, we first observe that
\[ [U_{1}, U_{2}] = 0 \Rightarrow [U_{1}^{k},U_{2}^{l}] = 0,  \; \forall \; k,l \in \mathbb{F}_{p}.\]
Therefore, if $U_{1}, U_{2} \in \cC$, then $U_{1}^{k}U_{2}^{l} \; \in \cC \; \forall \; k,l \in \mathbb{F}_{p}$. Furthermore, since $U_{1}^{p} = U_{2}^{p} = \cI_{p^{2}}$, this implies a cardinality of $p^{2}-1$ for the class $\cC$, as desired. 

It is easy to see that a maximal commuting class we do not need more than two operators $U_{1}, U_{2}$ to uniquely characterize a maximal commuting class. Suppose there exists $V \in \cC$ such that $V \neq U_1^{k}U_2^{l} ~ \forall~ k,l \in \mathbb{F}_p$. Since all integer powers modulo $p$ of $U_{1}, U_{2}, V$ would also commute, this would imply the class $\cC$ is of cardinality $p^{3} - 1$, which cannot exist in a space of dimension $d=p^{2}$.

We will often refer to such a set $\{U_{1},U_{2}, U_{1}U_{2}, U_{1}^{2}U_{2}, \ldots, U_{1}^{p-1}U_{2}\}$ of $p+1$ independent operators that give rise to a class $\cC$, as the {\bf generators} of the class $\cC$. Furthermore, since the class is completely determined once we pick a pair of independent, commuting operators $\{U_{1}, U_{2}\}$, we can represent $\cC$ in terms of a pair of generators as follows:
\[ \cC \equiv \langle U_1, U_2 \rangle .\]

\begin{property}\label{prop:two}
Every operator in a class commutes with exactly $p-1$ operators from another class. 
\end{property}
\begin{proof}
Consider a pair of mutually disjoint maximal commuting classes $\cC_{1}, \cC_{2}$ with generators $\{U_{1}, U_{2}\}$ and $\{V_{1}, V_{2}\}$ respectively:
\begin{eqnarray}
\cC_1 &\equiv& \langle U_1, U_2 \rangle ,\nonumber \\
\cC_2 &\equiv& \langle V_1, V_2 \rangle. \nonumber
\end{eqnarray}
We first show that any operator in $\cC_{1}$ must commute with atleast one operator in $\cC_{2}$. Suppose $U_1 \in \cC_{1}$ does not commute with any operator in $\cC_{2}$. Then the following commutation relations hold.
\[ [U_{1}, V_{1}] \neq 0 \; \Rightarrow U_{1}V_{1}  = \alpha V_{1}U_{1},\] where $\alpha$ is a $p^{th}$ root of unity. Similarly,
\[ [U_{1}, V_{2}] \neq 0 \; \Rightarrow U_{1}V_{2}  = \alpha^{j} V_{2} U_{1}, \] where $j \in \mathbb{F}_{p}$, and $j \neq 0$. These relations imply that,
\begin{equation*}
U_{1} (V_{1}^{k} V_{2}) = \alpha^{k+j} (V_{1}^{k} V_{2})U_{1} ~\forall~ k \in \mathbb{F}_{p}.
\end{equation*}
If $U_{1}$ does not commute with any element of $\cC_{2}$, we require $k+j \neq 0~mod~p~\forall~k\in \mathbb{F}_{p}$. This is not possible if $j \neq 0$. Hence $U_1 \in \cC_{1}$ must commute with at least one operator in $\cC_{2}$.

We may therefore assume without loss of generality that $[U_{1}, V_{1}] = 0$. This in turn implies that $[U_{1}, V_{1}^{k}]=0 ~ \forall k \in \mathbb{F}_p$. Hence, $U_{1} \in \cC_{1}$ commutes with $p-1$ operators from the class $\cC_2$.

Finally, we show that an operators in $\cC_{1}$ cannot commute with more than $p-1$ operators from another class. Suppose $U_{1}$ were to commute with another operator, say $V_{2}$ in $\cC_2$. This would imply $[U_{1},V_1^{k}V_2^{l}]=0 ~ \forall k,l \in \mathbb{F}_p$. In other words we would have an operator $U_{1} \in \cC_1$ which commutes with {\it all} operators in $\mathcal{C}_2$! But this would give rise to a set of commuting operators with cardinality greater $p^{2}-1$, which is not possible for operators on a $p^{2}$-dimensional space. Therefore, no operator can commute with more than $p-1$ operators in another class.
\end{proof}


In other words, every operator $U$ in a given class $\cC$ commutes with exactly one {\it independent} operator $V$ in another class $\cC'$ -- the remaining $(p-2)$ commuting operators in $\cC'$ are just powers of $V$. We also note two additional properties, which follow directly from Property~\ref{prop:two}.

\begin{property}\label{prop:three}
If $U_{1} \in \cC_{1}$, and $V_{1} \in \cC_{2}$ such that $[U_{1},V_{1}]=0$, then the operators $U_{1}^{k}V_{1}$ and $U_{1}^{l}V_{1}$ with $k \neq l \in \mathbb{F}_{p}$ must necessarily belong to different classes.
\end{property}
\begin{proof}
Note that $[U_{1}, V_{1}] = 0$ implies $[U_{1}, (U_1^k V_1)^{j}] = 0, \, \forall k,j \in \mathbb{F}_{p}$. Therefore, if $U_1^k V_1$ and $U_1^l V_1$ with $k \neq l$ were to belong to the same class $\cC$, the operator $U_{1} \in \cC_{1}$ would commute with $2(p-1)$ operators in the class $\cC$, in violation of Property~\ref{prop:two}. Hence, $U_1^{k} V_1$ and $U_1^{l} V_1$ must necessarily belong to two different classes.
\end{proof}

\begin{property}\label{prop:four}
Given two classes 
\[\cC_{1} \equiv \langle U_{1}, U_{2}\rangle, \; \cC_{2} \equiv \langle V_{1}, V_{2} \rangle, \] such that, 
\begin{eqnarray}
\left[U_{1}, V_{1}\right] &=& [U_{2}, V_{2}] = 0,  \nonumber \\ 
 \left[U_{1}, V_{2}\right]  \neq  0 &,&  [V_{1}, U_{2}] \neq 0 . \label{eq:comm_relations}  
\end{eqnarray}

Then, for a given $l \in \mathbb{F}_{p}$, there exists a unique $m \in \mathbb{F}_p$ such that $[U_1^{l}V_1\, , \,U_2^{m}V_2] = 0$.
\end{property}
\begin{proof}
Let $\omega = e^{2\pi i/p}$. Since $[V_{1}, U_{2}] \neq 0$, we may assume without loss of generality, that
\begin{eqnarray}
U_{1}V_{2} &=& \alpha V_{2}U_{1}, \; \alpha \in \{\omega, \omega^{2}, \ldots, \omega^{p-1}\}, \nonumber \\
V_{1}U_{2} &=& \alpha^{k} U_{2}V_{1}, \; k \in \mathbb{F}_{p}. \nonumber
\end{eqnarray}
Then we have, 
\[ (U_1^{l}V_1)(U_2^{m}V_2) = \alpha^{(km + l)} (U_{2}^{m}V_{2})(U_{1}^{l}V_{1}). \]
Thus, for the operators $U_1^{l}V_1$ and $U_2^{m}V_2$ to commute, we require, $km + l = 0~mod~p$. 

Finally, for a given pair $k,l \in \mathbb{F}_{p}$, we need to show that this expression holds true for a unique $m \in \mathbb{F}_p$. Let us assume that for a given $l$, $m$ is not unique. Therefore we have, $k m_1 + l = k m_2 + l = 0~mod~p$. Hence $m_1 = m_2~mod~p$, a contradiction, since  $m_1, m_2 \in \mathbb{F}_p$. Therefore $m_1 = m_2 = m \in \mathbb{F}_p$.
\end{proof}

\subsection{Existence and Construction of Unextendible MUBs}

Having noted a few basic properties of operator classes in prime-squared dimensions, we now proceed to discuss our construction of unextendible sets of classes in these dimensions. Our starting point will be a a partitioning of the unitary basis $\cU^{(p^{2})}$  into a complete set of $p^{2} +1$ mutually disjoint maximal commuting classes in $d=p^{2}$. We know that such a complete set always exists in prime-power dimensions from Lemma~\ref{lem:MUB_MCC} and the earlier results Wootters and Fields~\cite{WF89}. Starting with a complete set, we seek to identify subsets of classes whose elements might be used to construct newer classes. We first note that it suffices to restrict our attention to subsets of classes of cardinality $p+1$.

\begin{lemma}\label{lem:p+1}
Given a partitioning of the unitary basis $\cU^{(p^{2})}$ into a complete set of $p^{2} + 1$ classes, to form a new maximal commuting class of operators from $\cU^{(p^{2})}$, we require operators from exactly $p+1$ of the $p^{2} + 1$ classes. 
\end{lemma}
\begin{proof}
Consider a complete set $\Sigma \equiv \{\cC_{1}, \cC_{2}, \ldots, \cC_{p^{2}+1}\}$ of maximal commuting classes in dimension $d=p^2$. Now consider any new maximal commuting class $\mathcal{C}_I = \langle U_1, V_1 \rangle$ not belonging to $\Sigma$. We know from~\ref{prop:one} that there are $p+1$ independent operators characterizing the new class $\cC_{I}$. These $p+1$ operators must surely come from the classes belonging to $\Sigma$ since it is a complete partitioning of the unitary operators in $\cU^{p^{2}}$. We also know from~\ref{prop:two} that every operators in a class commutes with only one {\it independent} operators in a different class, and so each of the $p+1$ generators of $\mathcal{C}_I$ must come from different maximal commuting classes belonging
to $\Sigma$. Thus the new class $\cC_{I}$ is necessarily formed by picking $(p-1)$ operators each from $p+1$ classes.

More specifically, using ~\ref{prop:four} and the commutation relations in Eq.~\eqref{eq:comm_relations}, we can pick the $p-1$ generators of $\cC_{I}$ as follows: $U_1 \in \mathcal{C}_1$, $V_1 \in \mathcal{C}_2$, $U_1V_1 \in \mathcal{C}_3$,
$\ldots$, $U_1^{p-1}V_1 \in \mathcal{C}_{p+1}$, where $\mathcal{C}_1, \mathcal{C}_2, \ldots, \mathcal{C}_{p+1} \in \Sigma$.
\end{proof}

Suppose we do identify a set of $p+1$ classes $\{\cC_{1}, \cC_{2}, \ldots, \cC_{p+1}\}$ that belong to a complete set of classes, such that a new class $\cC_{I}$ can be formed using the operators in $\cup_{i=1}^{p+1}\cC_{i}$ , is it possible to form more classes using the same set of $p+1$ classes? This is answered in the following lemma. We merely state the result here and refer to the appendix~\ref{sec:appendix2} for the proof. 

\begin{lemma}\label{lem:atmost2}
No more than two new operator classes can be constructed from a set of $p+1$ classes belonging to a complete set of classes.
\end{lemma}

Once we have such a bound on the number of new classes that can be formed from a subset of the complete set if classes, we are naturally led to the following statement on the existence of unextendible sets of classes.
\begin{theorem} \label{th:UnextCons}
In dimensions $d=p^{2}$, there exist unextendible sets of classes of cardinality $N(p) = p^{2}-p+1$ or $N(p) = p^{2} - p + 2$, the common eigenbases of which form weakly unextendible sets of $N(p) = p^{2}-p+1$ or $N(p) = p^{2} - p + 2$ MUBs.
\end{theorem}
\begin{proof}
Properties \ref{lem:p+1} and \ref{lem:atmost2} imply that we can form either $0$, $1$ or $2$ new classes using a set of $p+1$ classes belonging to the complete set. Let $\{\cC_{1}, \cC_{2}, \ldots, \cC_{p+1}\}$ be a set of $p+1$ classes such that exactly one new class can be formed using the operators in $\cup_{i=1}^{p+1}$.  This new classes, together with the remaining set of $p^{2}-p$ classes ($\{\cC_{p+2}, \ldots, \cC_{p^{2}+1}\}$) is an unextendible set of $p^{2}-p+1$ classes.
Suppose $\{\cC_{1}, \cC_{2}, \ldots, \cC_{p+1}\}$ were a set of $p+1$ classes such that exactly $2$ new classes can be formed using the operators in $\cup_{i=1}^{p+1}$. Then, these new classes, together with the remaining set of $p^{2}-p$ classes ($\{\cC_{p+2}, \ldots, \cC_{p^{2}+1}\}$) form an unextendible set of $N(p) = p^{2} - p + 2$ classes.
\end{proof}

For example, consider the case of $p=3$, $d=3^{3}$. In $3^{2}$-dimensions, we can construct an unextendible set of  eight classes as illustrated in Fig.~\ref{fig:unext8}. Starting with a pair of classes $\cC_{1}, \cC_{2}$, we pick the remaining set of $7$ classes required to form a complete set, in such a way that there exist two new classes $\cC_{I}, \cC_{II}$ in $\cup_{i=1}^{4}\cC_{i}$. Then, $\cC_{I}, \cC_{II}$ along with $\{\cC_{5}, \ldots, \cC_{9}\}$ is an unextendible set of eight classes. 

\subsection{The case of $p=3$}

We can further restrict the cardinality of the unextendible sets in dimension $d=3^{2}$, using certain additional properties which hold in this case. We state and prove these additional properties in Appendix~\ref{sec:appendix3}, leading to the following result.



\begin{theorem} \label{th:p3case}
In $d=3^{2}$, consider a set of four classes $\cC_{1}$, $\cC_{2}$, $\cC_{3}$ and $\cC_{4}$ that belongs to a complete set of classes. If one new class can be constructed using the operators in $\cup_{i=1}^{4}\cC_{i}$, then, it is possible to construct one more class using the same set of four classes. Therefore, the cardinality of an unextendible set of classes $N(3) \neq 7$ in $d=3^{2}$. 
\end{theorem}
In other words, either (a) it is possible to find exactly two more classes $\cC_{I}, \cC_{II}$ using the operators in $\cup_{i=1}^{4}\cC_{i}$ giving rise to an unextendible set of eight classes (as shown in Fig.~\ref{fig:unext8}), or, (b) no new classes can be formed using the operators in $\cup_{i=1}^{4}\cC_{i}$.

%

We further show that if a set of four classes $\cC_{1}, \cC_{2}, \cC_{3}, \cC_{4}$ belongs to an unextendible set of $8$ classes in $d=3^{2}$, then, it is possible to form exactly one more class $\cC_{III}$ using the operators in $\cup_{i=1}^{4}\cC_{i}$. This implies the existence of an unextendible set of $5$ classes in $d=3^{2}$, as shown in Fig.~\ref{fig:unext5}.

\begin{figure}
    \centering
    \subfloat[Complete set of classes]{{\includegraphics[scale = 0.4]{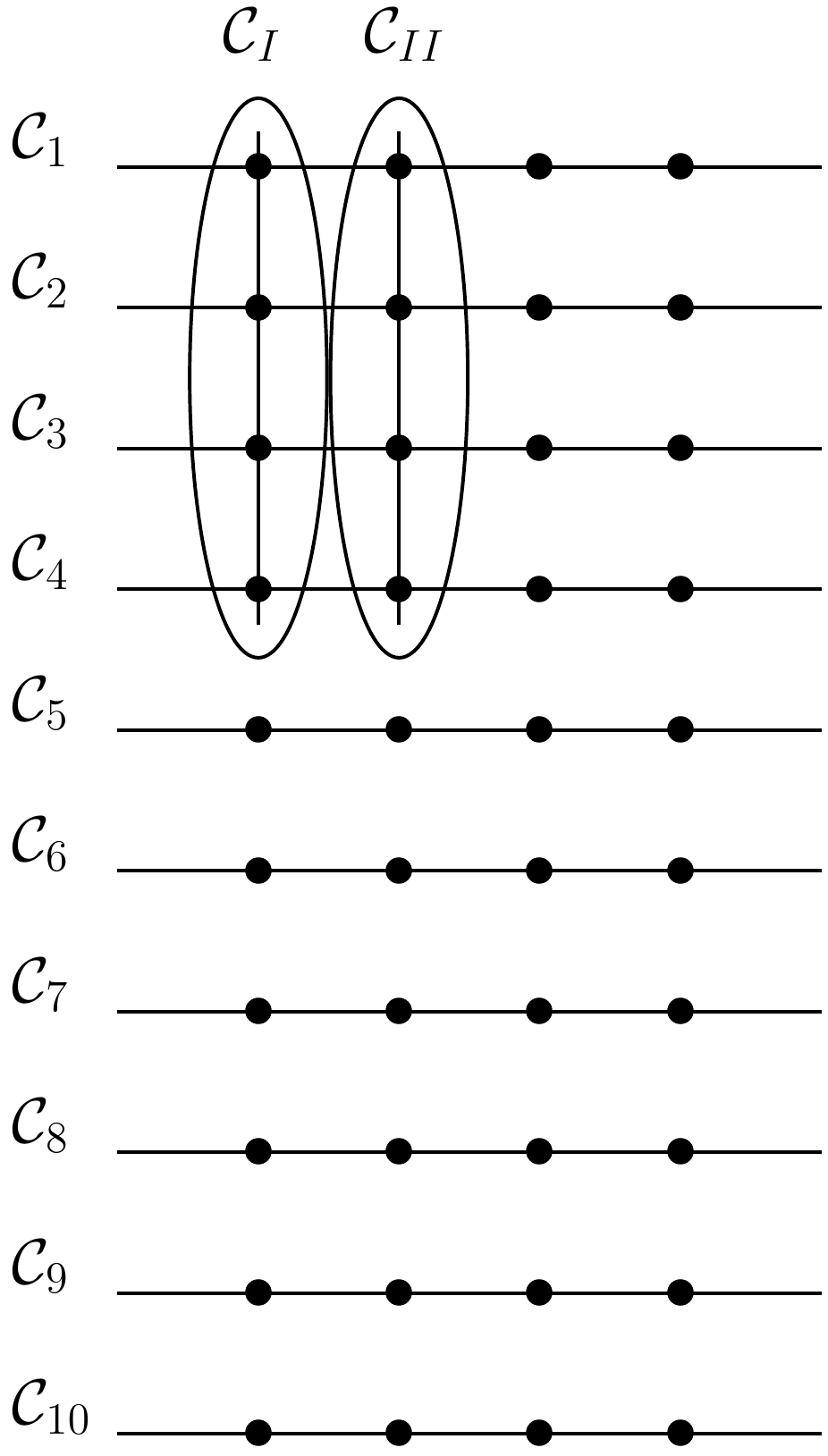} }}%
    \qquad
    \subfloat[Unextendible set of classes]{{\includegraphics[scale = 0.4]{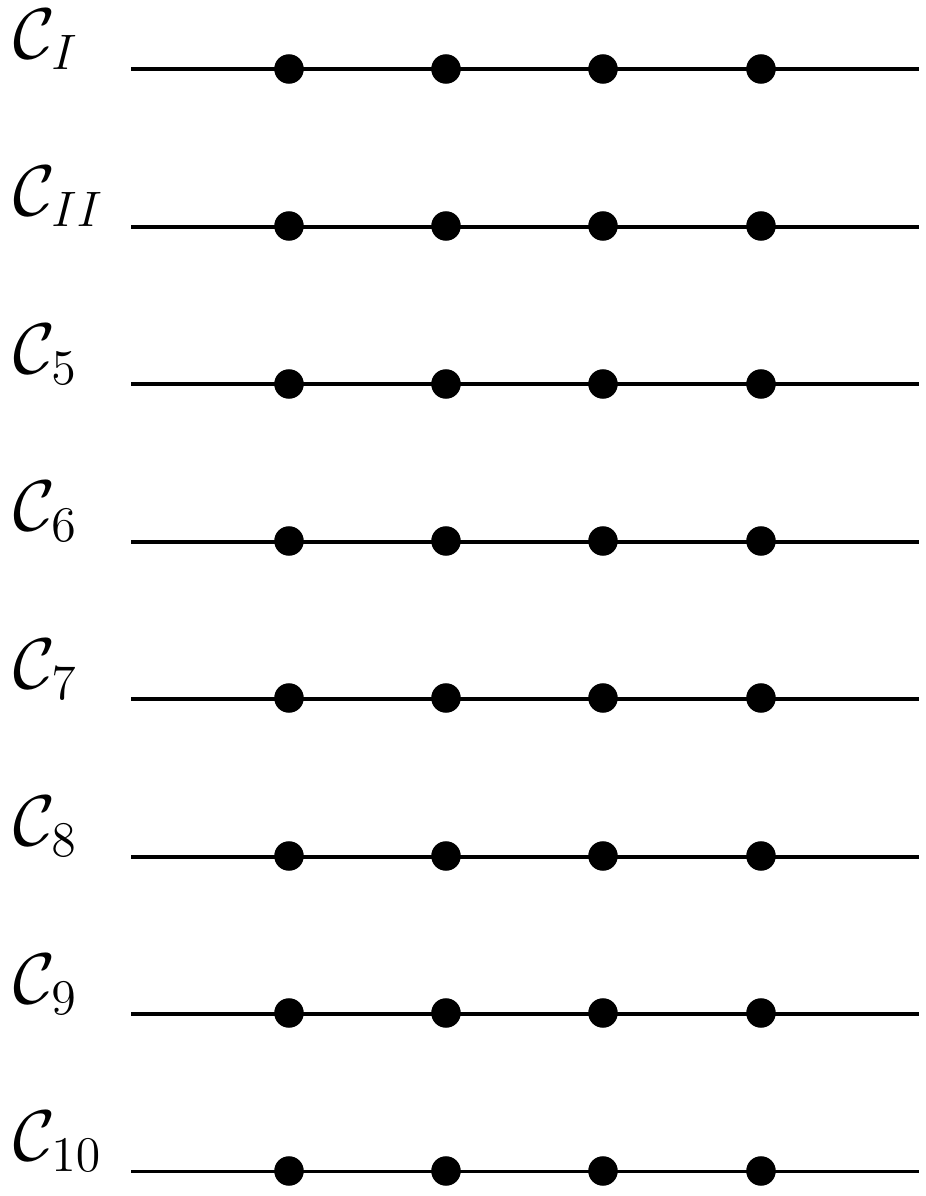} }}%
    \caption{Construction of unextendible set of $8$ classes in $d = 3^{2}$. Dots represent operators and horizontal lines represent classes. Vertical lines represent new classes constructed using existing classes. $\cC_{I}$ and $\cC_{II}$ together with $\cC_{5} - \cC_{10}$ forms an unextendible set of $8$ classes.}\label{fig:unext8}
\end{figure}
\begin{figure}
    \centering
    \subfloat[Unextendible set of $8$ classes]{{\includegraphics[scale = 0.4]{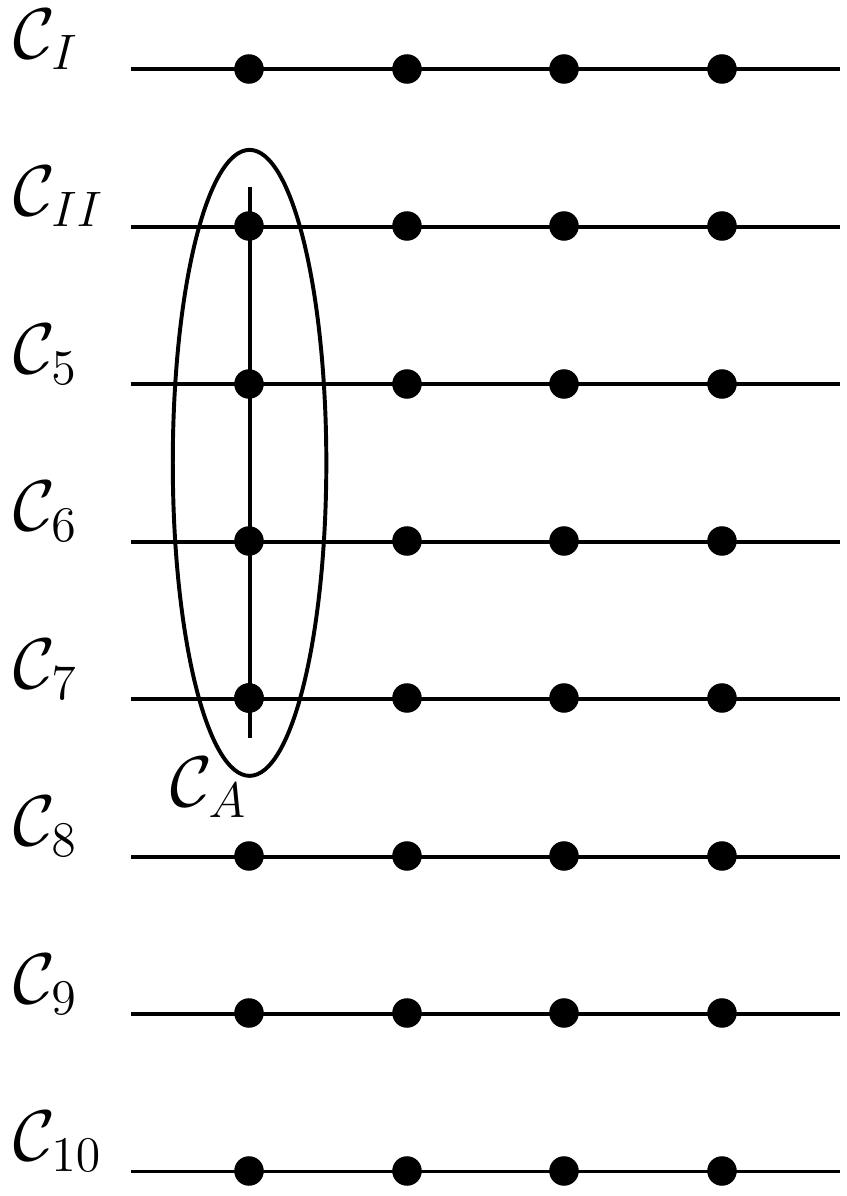} }}%
    \qquad
    \subfloat[Unextendible set of $5$ classes]{{\includegraphics[scale = 0.4]{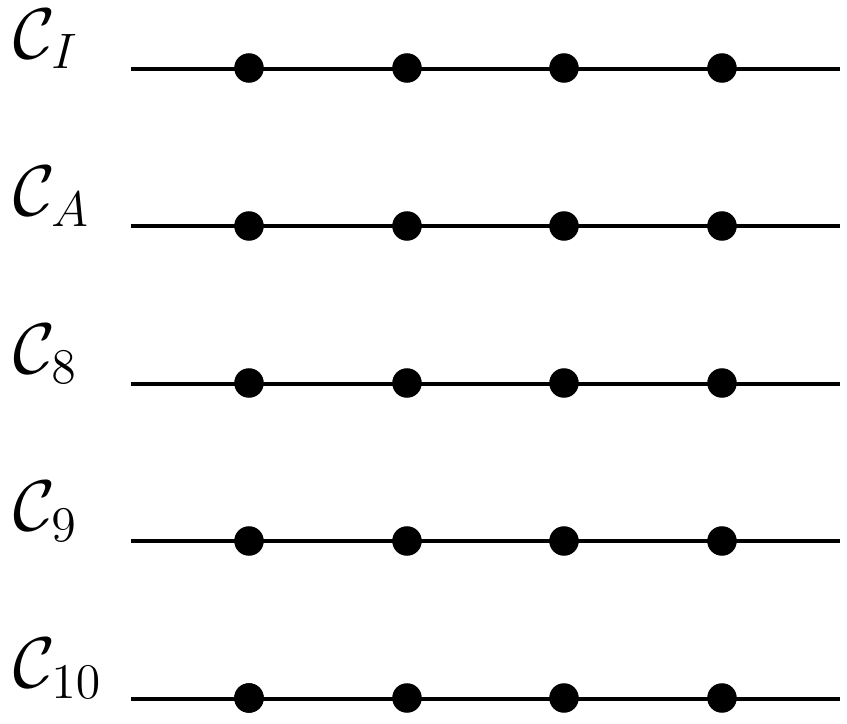} }}%
    \caption{Construction of unextendible set of $5$ classes in $d = 3^{2}$ using operators from the unextendible set of $8$ classes. Dots represent operators and horizontal lines represent classes. Vertical lines represent new classes constructed using existing classes. $\cC_{I}$ and $\cC_{A}$ together with $\cC_{8} - \cC_{10}$ forms an unextendible set of $5$ classes.}\label{fig:unext5}
\end{figure}

\section{Tightness of EURs in prime-squared dimensions}~\label{sec:tightEUR}

In this section, we describe an interesting connection between the existence of unextendible sets of classes in quantum systems of prime-squared dimensions and the tightness of two well known EURs in these dimensions. Indeed, our approach of constructing unextendible sets of classes leads to a systematic way of identifying MUBs that saturate both a Shannon ($H_{1}$) entropic and a $H_{2}$-entropic lower bound in prime-squared dimensions.  

An entropic uncertainty relation is a lower bound on the sum of the entropies associated with a set of measurement bases, measured {\it independently} on {\it identically prepared} copies of a quantum state. Recall that a measurement of basis $\cB^{(i)} \equiv \{|b^{(i)}_{j}\rangle\}$ in state $|\psi\rangle$ induces a probability distribution $p^{(i)}(j)_{|\psi\rangle} = \vert \langle b^{(i)}_{j}| \psi\rangle\vert^{2}$. Since the entropy is a measure of the {\it spread} of the distribution, we may use any valid entropic function $H(\{p^{(i)}(j)_{|\psi\rangle})\}$ to quantify the {\it uncertainty} in the outcome of the measurement $\cB^{(i)}$ in state $|\psi\rangle$. 

Here, we focus on two of the R\'enyi class of entropies~\cite{renyi:entropy}, namely the Shannon entropy ($H_{1}$, the R\'enyi entropy of order $1$) and the collision entropy ($H_{2}$, the R\'enyi entropy of order $2$), defined as:
\begin{eqnarray}
 H_{1}(\{p(i)\}) &:=& -\sum_{i}p(i)\log p(i), \nonumber \\
 H_{2}(\{p(i)\}) &:=& - \log\sum_{i}p(i)^{2} . \label{eq:entropies}
\end{eqnarray}
We denote the entropies associated with a measurement of basis $\cB^{(i)}$ on state $|\psi\rangle$ as $H_{\alpha}(\cB^{(i)}||\psi\rangle), \alpha = 1,2$. It is a well known that a pair of measurement bases $\cB^{(1)}, \cB^{(2)}$ in a $d$-dimensional system satisfy the following Shannon entropic uncertainty bound~\cite{maassen:ur}:
\begin{equation}
\frac{1}{2}\left[ H_{1}(\cB^{(1)}; |\psi\rangle) + H_{1}(\cB^{(2)}; |\psi\rangle)\right] \geq \frac{1}{2}\log d, \; \forall |\psi\rangle . \label{eq:MU}
\end{equation}
Furthermore, this bound is saturated iff the bases are mutually unbiased. As a trivial consequence of the above relation, we have a bound on the average Shannon entropy of {\it any} set of $L$ MUBs in $d$-dimensions (previously noted in~\cite{azarchs:entropy, BW07}):
\begin{equation}
\frac{1}{L}\sum_{i=1}^{L}H_{1}(\cB^{(i)}; |\psi\rangle) \geq \frac{1}{2}\log d, \; \forall |\psi\rangle. \label{eq:Shannon}
\end{equation}

It is known that there exist incomplete set of MUBs that saturate this weak entropic bound. In particular, it is known~\cite{BW07} that a set of $s+1$ MUBs in square dimensions $d=s^{2}$ which are realized by the action of product unitary operators on the computational basis saturate the bound in Eq.~\eqref{eq:Shannon}. 

Our construction of unextendible classes provides an alternate way of identifying sets of $p+1$ MUBs that saturate the EUR in Eq.~\eqref{eq:Shannon}, in prime-square dimensions ($d=p^{2}$). Furthermore, we also show that the same set of $p+1$ MUBs in $d=p^{2}$ also saturates the following $H_{2}$ entropic entropic relation, which was shown to hold for any set of $L$ MUBs in $d$-dimensions~\cite{WYM09,ww:urSurvey},
\begin{equation}
\frac{1}{L}\sum_{i=1}^{L}H_{2}(\cB^{(i)};|\psi\rangle) \geq -\log\left(\frac{L+d-1}{Ld}\right). \label{eq:H2}
\end{equation}
Indeed, the two lower bounds in Eqs.~\eqref{eq:Shannon} and~\eqref{eq:H2} coincide for the case of $L=p+1$ MUBs in $d=p^{2}$. 

\begin{theorem}\label{th:tightnessEUR}
In dimension $d = p^2$, consider a set of $(p+1)$ classes $\{\mathcal{C}_1, \mathcal{C}_2,\ldots, \mathcal{C}_{p+1}\}$, such that at least one more maximal commuting class $\mathcal{C_I}$ can be constructed by picking $(p-1)$ operators (of the form $U, U^2, \ldots, U^{p-1}$) from each of the classes. Then, the MUBs $\{\cB_{1}, \cB_{2}, \ldots, \cB_{p+1}\}$ corresponding to the classes $\{\cC_{1}, \cC_{2}, \ldots, \cC_{p+1}\}$ {\bf saturate} the following entropic uncertainty relations:
  \begin{eqnarray}
    \frac{1}{p+1} \sum_{i=1}^{p+1}H_2(\mathcal{B}_i||\psi \rangle) &\geq  \log_2  p , \nonumber \\
    \frac{1}{p+1} \sum_{i=1}^{p+1}H_{1}(\mathcal{B}_i||\psi \rangle) &\geq \log_2 p, \label{eq:EUR_bounds}
   \end{eqnarray}
with the lower bound attained by the common eigenstates of the newly constructed class $\cC_{I}$.
\end{theorem}
We provide the proof for this theorem in the appendix (\ref{sec:appendix4}). We may note that the states saturating the bound in Eq.~\eqref{eq:EUR_bounds} are indeed states that {\it look alike}~\cite{wootters:stateslooksame} with respect to each of the bases in the set $\{\cB_{1}, \cB_{2}, \ldots, \cB_{p+1}\}$.

\section{Concluding Remarks} 
We show by explicit construction the existence of weakly unextendible sets of MUBs of cardinalities $N(p) = p^{2}-p+1, p^{2} - p + 2$ in prime squared ($d=p^{2}$) dimensions. Our construction is based on grouping the generalized Pauli operators in these dimensions into sets of mutually disjoint, maximal commuting classes that are unextendible to a complete set of $(d+1)$ classes. We further demonstrate a general connection between the existence of unextendible sets and the tightness of entropic uncertainty relations for the $H_{1}$ and $H_{2}$ R\'enyi entropies.

Numerical evidence suggests that the MUBs obtained in our construction are in fact strongly unextendible. This is also borne out by a recent construction of unextendible sets of MUBs by exploring the connection between MUBs and complementary decompositions of maximal abelian subalgebras~\cite{szanto15}. Finally, it remains an interesting question to determine the cardinality of unextendible sets in prime-power dimensions ($d=p^{n}$). Since the techniques used here can be easily generalized to the case of prime-power dimensions, we may conjecture that the connection between the existence of unextendible sets and tightness of EURs can be extended to these dimensions as well.


\appendix

\section{Unextendible set of $8$ operator classes in $d=3^{2}$} \label{sec:appendix1}
Here is an example of a complete set of classes in $d = 3^2$ dimensions:
\begin{widetext}
    \begin{eqnarray}
    \mathcal{C}_{1} &=& \left\{ I\otimes \cX_{3}\cZ_{3}^{2},I\otimes \cX_{3}^{2}\cZ_{3},\cX_{3}\cZ_{3}\otimes I,\cX_{3}\cZ_{3}\otimes \cX_{3}\cZ_{3}^{2},\cX_{3}\cZ_{3}\otimes \cX_{3}^{2}\cZ_{3},\cX_{3}^{2}\cZ_{3}^{2}\otimes I,\cX_{3}^{2}\cZ_{3}^{2}\otimes \cX_{3}\cZ_{3}^{2},\cX_{3}^{2}\cZ_{3}^{2}\otimes \cX_{3}^{2}\cZ_{3}\right\} \nonumber \\
    \mathcal{C}_{2} &=& \left\{ \cZ_{3}\otimes \cX_{3}\cZ_{3}^{2},\cZ_{3}^{2}\otimes \cX_{3}^{2}\cZ_{3},\cX_{3}\otimes \cX_{3},\cX_{3}\cZ_{3}\otimes \cX_{3}^{2}\cZ_{3}^{2},\cX_{3}\cZ_{3}^{2}\otimes \cZ_{3},\cX_{3}^{2}\otimes \cX_{3}^{2},\cX_{3}^{2}\cZ_{3}\otimes \cZ_{3}^{2},\cX_{3}^{2}\cZ_{3}^{2}\otimes \cX_{3}\cZ_{3}\right\} \nonumber \\
    \mathcal{C}_{3} &=& \left\{ \cZ_{3}\otimes \cZ_{3}^{2},\cZ_{3}^{2}\otimes \cZ_{3},\cX_{3}\otimes \cX_{3}\cZ_{3},\cX_{3}^{2}\otimes \cX_{3}^{2}\cZ_{3}^{2},\cX_{3}\cZ_{3}^{2}\otimes \cX_{3}\cZ_{3}^{2},\cX_{3}^{2}\cZ_{3}\otimes \cX_{3}^{2}\cZ_{3},\cX_{3}^{2}\cZ_{3}\otimes \cX_{3}^{2}\cZ_{3},\cX_{3}\cZ_{3}^{2}\otimes \cX_{3}\cZ_{3}^{2}\right\} \nonumber \\
    \mathcal{C}_{4} &=& \left\{ I\otimes \cZ_{3},I\otimes \cZ_{3}^{2},\cX_{3}\otimes I,\cX_{3}\otimes \cZ_{3},\cX_{3}\otimes \cZ_{3}^{2},\cX_{3}^{2}\otimes I,\cX_{3}^{2}\otimes \cZ_{3},\cX_{3}^{2}\otimes \cZ_{3}^{2}\right\} \nonumber \\
    \mathcal{C}_{5} &=& \left\{ \cZ_{3}\otimes \cX_{3}^{2}\cZ_{3},\cZ_{3}^{2}\otimes \cX_{3}\cZ_{3}^{2},\cX_{3}\otimes \cX_{3}^{2},\cX_{3}\cZ_{3}\otimes \cX_{3}\cZ_{3},\cX_{3}\cZ_{3}^{2}\otimes \cZ_{3}^{2},\cX_{3}^{2}\otimes \cX_{3},\cX_{3}^{2}\cZ_{3}\otimes \cZ_{3},\cX_{3}^{2}\cZ_{3}^{2}\otimes \cX_{3}^{2}\cZ_{3}^{2}\right\} \nonumber \\
    \mathcal{C}_{6} &=& \left\{ I\otimes \cX_{3},I\otimes \cX_{3}^{2},\cX_{3}\cZ_{3}^{2}\otimes I,\cX_{3}\cZ_{3}^{2}\otimes \cX_{3},\cX_{3}\cZ_{3}^{2}\otimes \cX_{3}^{2},\cX_{3}^{2}\cZ_{3}\otimes I,\cX_{3}^{2}\cZ_{3}\otimes \cX_{3},\cX_{3}^{2}\cZ_{3}\otimes \cX_{3}^{2}\right\} \nonumber \\
    \mathcal{C}_{7} &=& \left\{ I\otimes \cX_{3}\cZ_{3},I\otimes \cX_{3}^{2}\cZ_{3}^{2},\cZ_{3}\otimes I,\cZ_{3}\otimes \cX_{3}\cZ_{3},\cZ_{3}\otimes \cX_{3}^{2}\cZ_{3}^{2},\cZ_{3}^{2}\otimes I,\cZ_{3}^{2}\otimes \cX_{3}\cZ_{3},\cZ_{3}^{2}\otimes \cX_{3}^{2}\cZ_{3}^{2}\right\} \nonumber \\
    \mathcal{C}_{8} &=& \left\{ \cZ_{3}\otimes \cZ_{3},\cZ_{3}^{2}\otimes \cZ_{3}^{2},\cX_{3}\otimes \cX_{3}^{2}\cZ_{3}^{2},\cX_{3}\cZ_{3}\otimes \cX_{3}^{2},\cX_{3}\cZ_{3}^{2}\otimes \cX_{3}^{2}\cZ_{3},\cX_{3}^{2}\otimes \cX_{3}\cZ_{3},\cX_{3}^{2}\cZ_{3}\otimes \cX_{3}\cZ_{3}^{2},\cX_{3}^{2}\cZ_{3}^{2}\otimes \cX_{3}\right\} \nonumber \\
    \mathcal{C}_{9} &=& \left\{ \cZ_{3}\otimes \cX_{3},\cZ_{3}^{2}\otimes \cX_{3}^{2},\cX_{3}\otimes \cX_{3}^{2}\cZ_{3},\cX_{3}\cZ_{3}\otimes \cZ_{3},\cX_{3}\cZ_{3}^{2}\otimes \cX_{3}\cZ_{3},\cX_{3}^{2}\otimes \cX_{3}\cZ_{3}^{2},\cX_{3}^{2}\cZ_{3}\otimes \cX_{3}^{2}\cZ_{3}^{2},\cX_{3}^{2}\cZ_{3}^{2}\otimes \cZ_{3}^{2}\right\} \nonumber \\
    \mathcal{C}_{10} &=& \left\{ \cZ_{3}\otimes \cX_{3}^{2},\cZ_{3}^{2}\otimes \cX_{3},\cX_{3}\otimes \cX_{3}\cZ_{3}^{2},\cX_{3}\cZ_{3}\otimes \cZ_{3}^{2},\cX_{3}\cZ_{3}^{2}\otimes \cX_{3}^{2}\cZ_{3}^{2},\cX_{3}^{2}\otimes \cX_{3}^{2}\cZ_{3},\cX_{3}^{2}\cZ_{3}\otimes \cX_{3}\cZ_{3},\cX_{3}^{2}\cZ_{3}^{2}\otimes \cZ_{3}\right\} \nonumber
    \end{eqnarray}
 We see that from $\mathcal{C}_{1}, \mathcal{C}_{2}, \mathcal{C}_{5}, \mathcal{C}_{7}$, we can form two new classes:
    \begin{eqnarray}
    \mathcal{C}_{I} &=& \left\{ I\otimes \cX_{3}\cZ_{3}^{2},I\otimes \cX_{3}^{2}\cZ_{3},\cZ_{3}\otimes \cX_{3}\cZ_{3}^{2},\cZ_{3}^{2}\otimes \cX_{3}^{2}\cZ_{3},\cZ_{3}\otimes \cX_{3}^{2}\cZ_{3},\cZ_{3}^{2}\otimes \cX_{3}\cZ_{3}^{2},\cZ_{3}\otimes I,\cZ_{3}^{2}\otimes I\right\} \nonumber \\
    \mathcal{C}_{II} &=& \left\{ \cX_{3}\cZ_{3}\otimes I,\cX_{3}^{2}\cZ_{3}^{2}\otimes I,\cX_{3}\cZ_{3}\otimes \cX_{3}^{2}\cZ_{3}^{2},\cX_{3}^{2}\cZ_{3}^{2}\otimes \cX_{3}\cZ_{3},\cX_{3}\cZ_{3}\otimes \cX_{3}\cZ_{3},\cX_{3}^{2}\cZ_{3}^{2}\otimes \cX_{3}^{2}\cZ_{3}^{2},I\otimes \cX_{3}\cZ_{3},I\otimes \cX_{3}^{2}\cZ_{3}^{2}\right\} \nonumber
    \end{eqnarray}
    Therefore $\cC_{I}$, $\cC_{II}$, $\cC_{3}$, $\cC_{4}$, $\cC_{6}$, $\cC_{8}$, $\cC_{9}$, $\cC_{10}$ form a set of
    $8$ unextendible maximal commuting classes.\\
     We make use operators from $4$ of the above $8$ classes (Where either $\cC_{I}$ or $\cC_{II}$ are necessarily used) to form a new class. This new class along with the remaining set of $4$ classes form a set of $5$ classes which are unextendible to a complete set of classes. In particular, we make use of the classes $\cC_{I}$, $\cC_{3}$, $\cC_{4}$ and $\cC_{8}$ to form a new class $\cC_{A}$ given by:
     \begin{eqnarray}
     \mathcal{C}_{A} &=& \left\{ I\otimes \cZ_{3}, \cZ_{3}\otimes I,\cZ_{3}\otimes \cZ_{3}, \cZ_{3}^{2}\otimes \cZ_{3}, I\otimes \cZ_{3}^{2},\cZ_{3}^{2}\otimes I,\cZ_{3}^{2}\otimes \cZ_{3}^{2},\cZ_{3}\otimes \cZ_{3}^{2}\right\} \nonumber 
     \end{eqnarray}
 The classes $\cC_{II}$, $\cC_{6}$, $\cC_{9}$, $\cC_{10}$ and $\cC_{A}$ form a set of $5$ classes which are unextendible to a complete set of classes as given below:
        \begin{eqnarray}
        \mathcal{C}_{II} &=& \left\{ \cX_{3}\cZ_{3}\otimes I,\cX_{3}^{2}\cZ_{3}^{2}\otimes I,\cX_{3}\cZ_{3}\otimes \cX_{3}^{2}\cZ_{3}^{2},\cX_{3}^{2}\cZ_{3}^{2}\otimes \cX_{3}\cZ_{3},\cX_{3}\cZ_{3}\otimes \cX_{3}\cZ_{3},\cX_{3}^{2}\cZ_{3}^{2}\otimes \cX_{3}^{2}\cZ_{3}^{2},I\otimes \cX_{3}\cZ_{3},I\otimes \cX_{3}^{2}\cZ_{3}^{2}\right\} \nonumber \\
         \mathcal{C}_{6} &=& \left\{ I\otimes \cX_{3},I\otimes \cX_{3}^{2},\cX_{3}\cZ_{3}^{2}\otimes I,\cX_{3}\cZ_{3}^{2}\otimes \cX_{3},\cX_{3}\cZ_{3}^{2}\otimes \cX_{3}^{2},\cX_{3}^{2}\cZ_{3}\otimes I,\cX_{3}^{2}\cZ_{3}\otimes \cX_{3},\cX_{3}^{2}\cZ_{3}\otimes \cX_{3}^{2}\right\} \nonumber \\
         \mathcal{C}_{9} &=& \left\{ \cZ_{3}\otimes \cX_{3},\cZ_{3}^{2}\otimes \cX_{3}^{2},\cX_{3}\otimes \cX_{3}^{2}\cZ_{3},\cX_{3}\cZ_{3}\otimes \cZ_{3},\cX_{3}\cZ_{3}^{2}\otimes \cX_{3}\cZ_{3},\cX_{3}^{2}\otimes \cX_{3}\cZ_{3}^{2},\cX_{3}^{2}\cZ_{3}\otimes \cX_{3}^{2}\cZ_{3}^{2},\cX_{3}^{2}\cZ_{3}^{2}\otimes \cZ_{3}^{2}\right\} \nonumber \\
         \mathcal{C}_{10} &=& \left\{ \cZ_{3}\otimes \cX_{3}^{2},\cZ_{3}^{2}\otimes \cX_{3},\cX_{3}\otimes \cX_{3}\cZ_{3}^{2},\cX_{3}\cZ_{3}\otimes \cZ_{3}^{2},\cX_{3}\cZ_{3}^{2}\otimes \cX_{3}^{2}\cZ_{3}^{2},\cX_{3}^{2}\otimes \cX_{3}^{2}\cZ_{3},\cX_{3}^{2}\cZ_{3}\otimes \cX_{3}\cZ_{3},\cX_{3}^{2}\cZ_{3}^{2}\otimes \cZ_{3}\right\} \nonumber \\
		\mathcal{C}_{A} &=& \left\{ I\otimes \cZ_{3}, \cZ_{3}\otimes I,\cZ_{3}\otimes \cZ_{3}, \cZ_{3}^{2}\otimes \cZ_{3}, I\otimes \cZ_{3}^{2},\cZ_{3}^{2}\otimes I,\cZ_{3}^{2}\otimes \cZ_{3}^{2},\cZ_{3}\otimes \cZ_{3}^{2}\right\} \nonumber 
    \end{eqnarray}
    \end{widetext}

\section{Construction of unextendible sets of operator classes in $d=p^{2}$} \label{sec:appendix2}

\begin{lemma}[\ref{lem:atmost2}]
No more than $2$ classes can be constructed using a set of $p+1$ classes belonging to a complete set of classes.
\end{lemma}
\begin{proof}
Consider a set of $p+1$ classes $\cC_{1}, \cC_{2}, \ldots, \cC_{p+1}$ using the elements of which, two more classes $\cC_{I}, \cC_{II}$ can be formed. Without loss of generality, we may assume the $p+1$ classes are of the following form: the classes $\cC_{1}$ and $\cC_{2}$ are first defined as $\cC_{1} \equiv \langle U_{1}, U_{2}\rangle$ and $\cC_{2} \equiv \langle V_{1}, V_{2}\rangle$, with generators that satisfy the commutation relations in Eq.~\eqref{eq:comm_relations}; the remaining $p-1$ classes are then constructed using the generators of $\cC_{1}, \cC_{2}$, as: 
\[\cC_{j} \equiv \langle U_{1}^{j-2}V_{1}, U_{2}^{j-2}V_{2}\rangle, j=3, \ldots, p+1 .\]
The two new classes $\cC_{I}$ and $\cC_{II}$ are then given by
\[ \cC_{I} \equiv \langle U_{1}, V_{1}\rangle, \; \cC_{II} \equiv \langle U_{2}, V_{2} \rangle\]

Let us now assume that we can form a third class $\cC_{III}$ using the operators in $\cup_{i=1}^{p+1}\cC_{i}$. We can assume without loss of generality that the generators  of $\cC_{III}$ are of the form, $U_1^{l}U_2$ and $V_1^{m}V_2$. Since they commute, we have $km+l = 0~mod~p$. We now need $(U_1^{l}U_2)^{r} V_1^{m}V_2  \in \mathcal{C}_j$ for some  $j \in \{3, 4, \ldots, p+1\}$ and $l, r, m \in \{1,2,\ldots, p-1\}$. As per our assumptions, the generators of the class $\mathcal{C}_j$ are $U_1^{j-2}V_1$ and $U_2^{s}V_2$, where $s$ of course depends on the value of $j$ and $s \in \{1,2,\ldots, p-1\}$. 

We only need to match the exponents of the individual components, that is, we need:
\begin{align*}
U_1^{lr}U_2^{r}V_1^{m}V_2 &= (U_1^{t}V_1)^{m}U_2^{s}V_2 \\
  &= U_1^{tm}U_2^{s}V_1^{m}V_2,
 \end{align*}
 where we have defined $t = j-2$, so that $t \in \{1,2,\ldots, p-1\}$. This in turn implies that $lr = tm~{\rm mod}~p$ and $r=s~{\rm mod}~p$, which holds only if,
\begin{itemize}
\item Case I: $r=t=s$ and $l=m$, or
\item Case II: $r=m=s$ and $l=t$
\end{itemize}
Consider Case I: We know that $U_1^{t}V_1$ and $U_2^{s}V_2$ commute. This immediately means that $ks + t = 0~mod~p$. Since $k \neq 0$, $t=s$ iff $s=t=0$, which is a contradiction because we know that $s,t \in \{1,2,\ldots, p-1\}$. Hence Case I is not possible.

Consider Case II: This would require $[U_1^{l}V_1, U_2^{m}V_2] = 0$ and $[U_1^{l}U_2, V_1^{m}V_2] = 0$. The latter commutation relation translates to $(p-k)m + l = 0~{\rm mod}~p$, while the former translates to $km + l = 0~{\rm mod}~p$, which is possible if and only if (a) $p$ is an even number, which leads a contradiction because we only consider primes $p>2$, or, (b) if $m=l=0$, which is again not possible because $m,l \in \{1,2,\ldots,p-1\}$.
Hence Case II is not possible.

This proves that we cannot form a third class from the set of $p+1$ classes mentioned before, for all prime $p > 2$.
\end{proof}
 
NOTE: For $p=2$, Case II is possible and hence another class can be formed, as already noted in~\cite{unext_MBGW}. 


\section{The case of $d=3^{2}$}\label{sec:appendix3}

We first observe two properties of a complete set of $p^{2}+1$ classes in $d=p^{2}$, which are not specific to $p=3$. 

\begin{property}\label{prop:eight}
Let $\cC_{1}, \cC_{2}, \ldots, \cC_{p+1}$ be a set of $p+1$ classes out of a complete set of classes in $d=p^{2}$, such that a new class $\cC_{I} \equiv \langle U_{1}, V_{1}\rangle $ can be formed using  $U_{1} \in \cC_{1}$, $V_{1} \in \cC_{2}$, $U_{1}V_{1} \in \cC_{3}$ , etc. Then, there exist $\tilde{U}_2 \in \cC_{1}$ and $\tilde{V}_{2} \in \cC_{2}$ such that $\tilde{U}_{2}\tilde{V}_{2} \in \cC_{3}$.
\end{property}
\begin{proof}
Let $\cC_{1} \equiv \langle U_{1}, U_{2}\rangle$ and $\cC_{2} \equiv \langle V_{1}, V_{2}\rangle$, with the generators satisfying the commutation relations in Eq.~\eqref{eq:comm_relations}. We are given that $U_{1}V_{1} \in \cC_{3}$. We need to find another independent operator in $\cC_{3}$ so as to completely generate $\cC_{3}$. 

Starting with any operator of the form $W = U_{1}^{a}V_{1}^{b}U_{2}^{c}V_{2}^{d}$, we have the following constraints on $a,b,c,d \in \mathbb{F}_{p}$ so that $W$ is a valid operator in $\cC_{3}$, given the classes $\cC_{1}$ and $\cC_{2}$. 
\begin{itemize}
\item[(i)] $[U_{2}, U_{1}^{a}V_{1}^{b}U_{2}^{c}V_{2}^{d}] = 0$. We impose this condition since we need to have an operator in
$\cC_{3}$ commuting with $U_{2}$. This implies that $b = 0~mod~p$. The operator $W$ is thus of the form $U_{1}^{a}U_{2}^{c}V_{2}^{d}$.
\item[(ii)] $[U_{1}V_{1}, U_{1}^{a}U_{2}^{c}V_{2}^{d}] = 0$. Therefore $(p-1)c + d = 0~mod~p$. Hence $c = d~mod~p$.
The operator $W$ is thus of the form $U_{1}^{a}U_{2}^{c}V_{2}^{c}$.
\end{itemize}
If $W = U_{1}^{a}U_{2}^{c}V_{2}^{c}$ is indeed a generator of $\cC_{3}$ all products of $U_{1}V_{1}$ and $U_{1}^{a}U_{2}^{c}V_{2}^{c}$ must belong in $\cC_{3}$. Hence 
\[(U_{1}^{a}U_{2}^{c}V_{2}^{c})^{2} (U_{1}V_{1})^{p-a} = U_{1}^{a}V_{1}^{p-a}U_{2}^{2c}V_{2}^{2c} \in \cC_{3}. \] The last operator is obtained as a product of
$U_{1}^{a}U_{2}^{2c} \in \cC_{1}$ and $V_{1}^{p-a}V_{2}^{2c} \in \cC_{2}$. Therefore, if we let,
\begin{align*}
\tilde{U}_{2}& := U_{1}^{a}U_{2}^{2c}\\
\tilde{V}_{2}& := V_{1}^{p-a}V_{2}^{2c}
\end{align*}
we have $\tilde{U}_{2}\tilde{V}_{2} \in \cC_{3}$ as desired.
\end{proof}

\begin{property}\label{prop:nine}
Consider $\cC_{1}, \cC_{2}, \ldots, \cC_{p+1}$ to be part of a complete set of classes in $d=p^{2}$ system such that at least one new class can be formed by picking $p-1$ operators from each of these $p+1$ classes. Let the generators of $\cC_{1} \equiv \langle U_{1}, U_{2}\rangle$ and $\cC_{2} \equiv \langle V_{1}, V_{2}\rangle$, with $U_{1}V_{2} = \alpha V_{2} U_{1}$. Then, we can always redefine $U_{2} \in \cC_{1}$ such that $U_{2}V_{1} = \alpha V_{1}U_{2}$, where $\alpha$ is a $p^{th}$ root of unity.
\end{property}
\begin{proof}
Let us suppose that $U_{2}V_{1} = \alpha^{j} V_{1}U_{2}$ for some $j \in \mathbb{F}_{p} \setminus \{0\} $. We know that $U_{2}^{k} \in
\cC_{1} ~\forall~ k \in \mathbb{F}_{p}$. Therefore $U_{2}^{j^{-1}} \in \mathbb{F}_{p}$. Letting $U_{2} := U_{2}^{j^{-1}}$, we see that
$U_{2}V_{1} = \alpha V_{1}U_{2}$.
\end{proof}

\begin{theorem*}[\ref{th:p3case}]
In $d = 3^{2}$, given $\cC_{1}$, $\cC_{2}$, $\cC_{3}$ and $\cC_{4}$ belonging to a complete set of classes and using which
one new class can be constructed, it is always possible to construct another class using the set of $4$ classes. In other words, it is not possible to construct $N(3) = 3^{2}-3+1 = 7$ unextendible set of classes.
\end{theorem*}
\begin{proof}
As before, we assume that the classes $\cC_{1}, \cC_{2}$ are of the form 
\[\cC_{1} \equiv \langle U_{1}, U_{2}\rangle, \; \cC_{2} \equiv \langle V_{1}, V_{2}  \rangle .\]
Further, we assume that $U_{1}V_{2} \in \cC_{3}$ and $U_{1}^{2}V_{1} \in \cC_{4}$, implying the existence of one more class $\cC_{I} \equiv \langle U_{1}, V_{1} \rangle$. We also know from properties~\ref{prop:eight} and~\ref{prop:nine} that the generators of $\cC_{1}, \cC_{2}$ can be chosen to satisfy $U_{1}V_{2} = \alpha V_{2}U_{1}$, $U_{2}V_{1} = \alpha V_{1}U_{2}$, such that  $U_{2}V_{2} \in \cC_{3}$.

Let us now assume that it is not possible to find one more class using operators in the set $\{\cC_{1}, \cC_{2}, \cC_{3}, \cC_{4}\}$. This would imply that $U_{2}^{2}V_{2} \notin \cC_{4}$, and therefore, $U_1^2V_1 U_2^2V_2 \notin \cC_4$. Since we are given that the set $\{\cC_{1}, \cC_{2}, \cC_{3}, \cC_{4}\}$ is part a complete set of classes, the operator $U_1^2V_1 U_2^2V_2$ must belong some other class, say, $\cC_{n}, 5 \leq n \leq p^(2)+1$.


We then show that it is not possible to complete the class $\cC_{n}$, in such a way that it is mutually disjoint with  $\{\cC_{1}, \cC_{2}, \cC_{3}, \cC_{4}\}$. We first note that a general operator $G = U_1^{a}V_1^{b}U_2^{c}V_2^{d} \in \cC_{n}$ must satisfy the following conditions:

\begin{eqnarray}
\left[ U_1^2V_1U_2^2V_2 , U_1^{a}V_1^{b}U_2^{c}V_2^{d} \right] &=& 0, \\
\left[ U_1^2V_1 , U_1^{a}V_1^{b}U_2^{c}V_2^{d} \right] &\neq& 0 \\
\left[ U_2^2V_2 , U_1^{a}V_1^{b}U_2^{c}V_2^{d} \right] &\neq& 0 \\
\left[ U_1^2U_2, U_1^{a}V_1^{b}U_2^{c}V_2^{d} \right] &\neq& 0 \\ 
\left[ V_1^2V_2, U_1^{a}V_1^{b}U_2^{c}V_2^{d} \right] &\neq& 0 
\end{eqnarray}

Eq.(1) along with inequalities (2) and (3) has two possible solutions: (i) $b=2a+1$ and $d = 2c+2$, or, (ii) $b = 2a+2$ and $d = 2c+1$.

Option (i) along with inequalities (4),(5) above imply that $c = (a+2){\rm mod} 3$. This implies that the second generator $G$ of the class $\cC_n$ must be of the following form: $U_1^{a}V_1^{2a+1}U_2^{a+2}V_2^{2a}$. Substituting different values of $a \in \mathbb{F}_{3}$,  we get,
\[ G \in \{ U_1V_2^2, V_1U_2^2 , U_1^2V_1^2U_2V_2\}. \] We We note that option (ii), along with inequalities (4), (5) do not give any further solutions: they simply yield the squares of these operators. 

It is easy to see that the set of solutions for $G$, along with the operator $U_{1}^{2}V_{1}U_{2}^{2}V_{2}$ is closed under multiplication. Therefore, all three possible solutions for the second generator lead to the same class: 
\[ \cC_{n} \equiv \langle U_1^2 V_1 U_2^2 V_2, U_1^2 V_1^2 U_2V_2 \rangle\].  
However, we have chosen the operators $U_{2}, V_{2}$  that $U_1^2 V_1^2 U_2V_2 \in \cC_3$. Thus the class $\cC_{n}$ containing the operator $U_1^2 V_1 U_2^2 V_2$ cannot be  mutually disjoint with the initial set of three classes. Therefore, either $U_1^2 V_1 U_2^2 V_2 \in \cC_{4}$, leading to a {\it second} class $\cC_{II} \equiv \langle U_{2}, V_{2} \rangle$, or, $\cC_{1}, \cC_{2}, \cC_{3}, \cC_{4}$ cannot be extended to a complete set of $p^{2}+1$ classes.
\end{proof}

\section{Proof of tightness of EURs} \label{sec:appendix4}
  
A key ingredient of the proof pf Theorem~\ref{th:tightnessEUR} is a parameterization of the basis-vectors of the bases $\{\cB_i\}$ corresponding to classes $\{\cC_{i}\}$, in terms of vectors in a $p+1$-dimensional vector space over the field $\mathbb{F}_{p}$. We first index the operators in class $\cC_{i}$ as follows: starting with a pair of generators $ \sigma_{i}^{(1)}, \sigma_{i}^{(2)}$, we may obtain a set of $p+1$ {\it independent} operators as products of these two generators:
  \[ \sigma_{i}^{(k)} = (\sigma_{i}^{(1)})^{k-2}(\sigma_{i}^{(2)}) ~ \forall~k \in [3,p+1].\]
The remaining operators in $\cC_{i}$ are simply powers of $\sigma_{i}^{(k)}, k \in [1,p+1]$. A general operator in $\cC_{i}$ is thus denoted as $(\sigma_{i}^{(k)})^{j}, j \in \mathbb{F}_{p}$.

Consider an operator in the span of the operators constituting the class $\cC_{i}$, of the following form:
\begin{eqnarray}
\rho_{i}^{\bf x}  &=& \frac{I}{p^2} + \sum_{l=1}^{p+1}\left[\omega ^ {x_{l}} \sigma_{i}^{(l)} + \omega ^ {2x_l} (\sigma_{i}^{(l)} )^2\right] \nonumber \\
&+& \ldots  + \sum_{l=1}^{p+1}\left[\omega ^ {(p-1)x_l} (\sigma_{i}^{(l)} )^{p-1} \right], \label{eq:basis_vec}
\end{eqnarray}
where, ${\bf x} \equiv (x_{1}, x_{2}, \ldots, x_{p+1}) \in \mathbb{F}_{p}^{p+1}$, with 
\[ 	x_{k} = (k-2)x_{1} + x_{2}~mod~p , \; k \in [3,p+1] .\]
Clearly, $Tr(\rho_{i}^{\bf x}) = 1$, and, $(\rho_{i}^{\bf x})^{\dagger} = \rho_{i}^{\bf x}$. The latter follows from the fact that, \[ (\sigma_{i}^{(l)})^{\dagger} = (\sigma_{i}^{(l)})^{p+1}. \]
Further, it is easy to check that $Tr((\rho_{i}^{\bf x})^{2}) = 1$:
\begin{align*}
Tr((\rho_{i}^{\bf x})^{2}) &= \frac{1}{p^{4}}Tr(I +  \sum_{l=1}^{p+1} \sum_{k=1}^{p-1} I).\\
&= \frac{1}{p^{4}} Tr(I + (p^{2}-1)I) = \frac{1}{p^{2}}Tr(I) = 1.
\end{align*}

Finally, the following lemma proves that $\rho_{i}^{\bf x}$ is in fact a rank-$1$ projector. 
\begin{lemma}
The operator defined in Eq.~\eqref{eq:basis_vec} satisfies $(\rho_{i}^{\bf x})^{2} = \rho_{i}^{\bf x}$.
\end{lemma}
\begin{proof}
\begin{widetext}
\begin{eqnarray}
&& (\rho_{i}^{\bf x})^{2} = \frac{1}{p^{2}}(\rho_{i}^{\bf x} + (\omega^{x_{1}} \sigma_{i}^{(1)}\rho_{i}^{\bf x} + \omega^{x_{2}}\sigma_{i}^{(2)}\rho_{i}^{\bf x} + \ldots + \omega^{x_{p+1}}\sigma_{i}^{(p+1)}\rho_{i}^{\bf x}) \nonumber \\ 
&+& (\omega^{2x_{1}}(\sigma_{i}^{1})^{(2)} \rho_{i}^{\bf x} + \omega^{2x_{2}} (\sigma_{i}^{(2)})^{2} \rho_{i}^{\bf x} + \ldots + \omega^{2x_{p+1}} (\sigma_{i}^{(p+1)})^{2} \rho_{i}^{\bf x} ) \nonumber \\
&+& \ldots \nonumber \\ 
&+& (\omega^{(p-1)x_{1}} (\sigma_{i}^{(1)})^{p-1} \rho_{i}^{\bf x} + \omega^{(p-1)x_{2}}(\sigma_{i}^{(2)})^{p-1} \rho_{i}^{\bf x} + \ldots + \omega^{(p-1)x_{p+1}} (\sigma_{i}^{(p+1)})^{p-1}  \rho_{i}^{\bf x}) ). \label{eq:vec_sum}
\end{eqnarray}
\end{widetext}
We require that
\[ (\rho_{i}^{\bf x})^{2} = \frac{1}{p^{2}} (p^{2}\rho_{i}^{\bf x}) = \rho_{i}^{\bf x} .\]
This would imply that we need each term of the summation in Eq.~\eqref{eq:vec_sum} be $\rho_{i}^{\bf x}$. We show that this is indeed the case, since 
\[ \omega^{ax_{b}}(\sigma_{i}^{(b)})^{a}\rho_{i}^{\bf x} = \rho_{i}^{\bf x} ~ \forall \; a \in \mathbb{F}_{p}, \forall b \in [1,p+1].\]
To see this, consider a generic term in the operator sum $\omega^{ax_{b}}(\sigma_{i}^{b})^{a}\rho_{i}^{\bf x}$ of the form,
\[ \omega^{ax_{b}}(\sigma_{i}^{(b)})^{a} \times \omega^{cx_{d}}(\sigma_{i}^{(d)})^{c} = \omega^{ax_{b} + cx_{d}} (\sigma_{i}^{(b)})^{a}(\sigma_{i}^{(d)})^{c}.\]
We have the following:
\begin{widetext}
 \begin{align*}
ax_{b} + cx_{d} &= a((b-2(1-\delta_{b,1}))x_{1} + (1-\delta_{b,1})x_{2}) +  c((d-2(1-\delta_{d,1}))x_{1} + (1-\delta_{d,1})x_{2}).\\
		&= (ab + cd - 2(a(1-\delta_{b,1})+c(1-\delta_{d,1})))x_{1} + (a(1-\delta_{b,1})+c(1-\delta_{d,1}))x_{2}.\\
		(\sigma_{i}^{(b)})^{a}(\sigma_{i}^{(d)})^{c} &= ((\sigma_{i}^{1})^{(b-2)}(\sigma_{i}^{(2)})^{1-\delta_{b,1}})^{a}((\sigma_{i}^{(1)})^{d-2}(\sigma_{i}^{(2)})^{1-\delta_{d,1}})^{c}.\\
		&= (\sigma_{i}^{(1)})^{ab+cd - 2(a(1-\delta_{b,1})+c(1-\delta_{d,1}))}(\sigma_{i}^{(2)})^{a(1-\delta_{b,1})+c(1-\delta_{d,1})}.
\end{align*}
\end{widetext}
We now have two possibilities:
\begin{itemize}
\item[(\bf{A}) ] $a(1-\delta_{b,1})+c(1-\delta_{d,1}) = 0~mod~p$. Therefore 
\begin{align*}
ax_{b} + cx_{d} &= (ab + cd)x_{1} = \gamma x_{1}, \\
(\sigma_{i}^{(b)})^{a}(\sigma_{i}^{(d)})^{c} &= (\sigma_{i}^{(1)})^{ab+cd} 	= (\sigma_{i}^{(1)})^{\gamma},
\end{align*}
where, $\gamma = (ab+cd) ~ mod~p$.
\item[(\bf{B}) ] $a(1-\delta_{b,1})+c(1-\delta_{d,1}) \neq 0~mod~p$, in which case, $(a(1-\delta_{b,1})+c(1-\delta_{d,1}))^{-1}$ exists. Therefore we have,
\begin{align*}
ax_{b} + cx_{d} &= (a(1-\delta_{b,1})+c(1-\delta_{d,1}))(((ab+cd) \\& (a(1-\delta_{b,1})+c(1-\delta_{d,1}))^{-1} - 2)x_{1} + x_{2}).\\
					&= \gamma x_{\beta}.\\
					(\sigma_{i}^{(b)})^{a}(\sigma_{i}^{(d)})^{c} &= ((\sigma_{i}^{(1)})^{\beta - 2}\sigma_{i}^{(2)})^{\gamma} = (\sigma_{i}^{(\beta}))^{\gamma}
\end{align*}
where $\gamma = (a(1-\delta_{b,1})+c(1-\delta_{d,1}))~mod~p$ and $\beta = (ab+cd)(a(1-\delta_{b,1})+c(1-\delta_{d,1}))^{-1} ~ mod~(p+1)$.
\end{itemize}
To summarize, for $a,c \in \mathbb{F}_{p}$ and $b,d \in [ 1, p+1 ]$ we have shown, 
\[
\omega^{ax_{b}}(\sigma_{i}^{(b)})^{a} \times \omega^{cx_{d}}(\sigma_{i}^{(d)})^{c} = \omega^{\gamma x_{\beta}}(\sigma_{i}^{(\beta)})^{\gamma},
\]
for some $\gamma \in \mathbb{F}_{p}$ and $\beta \in [1,p+1]$. Thus, $\omega^{\gamma x_{\beta}}(\sigma_{i}^{(\beta)})^{\gamma}$ is one of the terms that occurs when $\rho_{i}^{\bf x}$ is expanded in the $\{\sigma^{(i)}\}$ operator basis. Every term in the operator expansion for the product of $\rho_{i}^{\bf x}$ with $\omega^{ax_{b}}(\sigma_{i}^{(b)})^{a}$ gives a unique term in the operator expansion of $\rho_{i}^{\bf x}$. We therefore conclude that $\omega^{ax_{b}}(\sigma_{i}^{(b)})^{a}\rho_{i}^{\bf x} = \rho_{i}^{\bf x} ~ \forall a \in \mathbb{F}_{p}, \forall b \in [1,p+1]$. 
\end{proof}

$\rho_{i}^{\bf x}$ is therefore a valid pure state in the span of the operators belonging to the $i^{th}$ class. We may therefore rewrite the state as $\rho_{i}^{\bf x} = |b_{i}^{\bf x}\rangle \langle b_{i}^{\bf x}|$, where, $|b_{i}^{\bf x}\rangle$ is a common eigenvector of the operators belonging to the class $\cC_{i}$. Recall that the vector ${\bf x} \in \mathbb{F}_{p}^{p+1}$ is determined when we pick a pair of elements $x_{1}, x_{2} \in \mathbb{F}_{p}$. Since $x_{1}$ and $x_{2}$ can each take on $p$ values, we can find $p^{2}$ such pure states in the space of the operators belonging to a given class $\cC_{i}$. Furthermore, we show that these $p^{2}$ states are indeed orthogonal to each other, thus constituting an orthonormal basis for the class $\cC_{i}$. 

Consider states $|b_{i}^{\bf x}\rangle, |b_{i}^{\bf y}\rangle$, with ${\bf x} \neq {\bf y}$. Then,
\begin{eqnarray}
&& \tr[\, |b_i^{\bf x}\rangle \langle b_i^{\bf x}|b_i^{\bf y}\rangle \langle b_i^{\bf y}| \,] \nonumber \\
&=& \frac{1}{p^2} + \frac{1}{p^{2}} \sum_{l=1}^{p+1} [ \omega ^ {x_{l} + (p-1)y_{l}} + \omega ^{2x_{l} + (p-2)y_{l}}] \nonumber \\
& +& \frac{1}{p^{2}}[\, \ldots + \sum_{l=1}^{p+1}\omega ^ {(p-1)x_{l} + y_{l}} ].
\end{eqnarray}
For the RHS to vanish, we require that $x_{l} = y_{l}$ for one and only one $l \in \{1,2,\ldots,p+1\}$. In other words, the strings given by ${\bf x} = x_{1}x_{2}\ldots x_{p+1}$ and ${\bf y} = y_{1}y_{2}\ldots y_{p+1}$ must agree at exactly one position, in order for the corresponding states $|b_{i}^{\bf x}\rangle$, $|b_{i}^{\bf y}\rangle$ to be mutually orthogonal. Given a vector ${\bf x}$, there exist exactly $p^{2}$ vectors that coincide with ${\bf x}$ at exactly one location. Thus, corresponding to each class $\cC_{i}$, we have an orthonormal basis for the entire Hilbert space $\mathbb{C}^{p^{2}}$, with basis vectors 
\begin{eqnarray}
|b_i^{\bf x}\rangle \langle b_i^{\bf x}|  &=& \frac{I}{p^2} + \frac{1}{p^{2}}\sum_{l=1}^{p+1}\left[ \omega ^ {x_{l}} \sigma_{i}^{l} + \omega ^ {2x_l} (\sigma_{i}^{l} )^2 \right] \nonumber \\
&+&     \ldots + \frac{1}{p^{2}}\sum_{l=1}^{p+1}\omega ^ {(p-1)x_l} (\sigma_{i}^{l} )^{p-1} . \label{eq:basis_state}
\end{eqnarray}
We are now ready to prove Theorem~\ref{th:tightnessEUR}, the statement of which we recall here.

\begin{theorem*}[\ref{th:tightnessEUR}]
   In dimension $d = p^2$, consider a set of $(p+1)$ classes $\{\mathcal{C}_1, \mathcal{C}_2,\ldots, \mathcal{C}_{p+1}\}$,
   such that at least one more maximal commuting class $\mathcal{C_I}$ can be constructed by picking $(p-1)$ operators
   (of the form $U, U^2, \ldots, U^{p-1}$) from each of the classes. Then, the MUBs $\{\cB_{1}, \cB_{2}, \ldots, \cB_{p+1}\}$ corresponding to the classes $\{\cC_{1}, \cC_{2}, \ldots, \cC_{p+1}\}$ saturate the following entropic uncertainty relations:
  \begin{eqnarray}
    \frac{1}{p+1} \sum_{i=1}^{p+1}H_2(\mathcal{B}_i||\psi \rangle) &\geq \log  p , \label{eq:H2p} \\
    \frac{1}{p+1} \sum_{i=1}^{p+1}H_{1}(\mathcal{B}_i||\psi \rangle) &\geq \log  p, \label{eq:H1p}
 \end{eqnarray}
with the lower bound attained by the common eigenstates of the newly constructed class $\cC_{I}$.
 \end{theorem*}

\begin{proof}
Let $\cC_{1}, \cC_{2}, \ldots, \cC_{p+1}$ be a set of $p+1$ classes such that at least one maximal commuting class can be constructed by picking $p-1$ operators from each class. We denote by $\{|b_i^{\bf a_m}\rangle\}$ the states that constitute the basis $\cB_i$ associated with the class $\cC_{i}$, using a set of  vectors $\mathbf{a_m} \in (\mathbb{F}_{p})^{p+1}$, $m \in [1,p^2]$. Thus $|b_i^{\bf a_m}\rangle$ is the $m^{th}$ basis vector of the $i^{th}$ basis. 

Now consider a class constructed from $\{\mathcal{C}_1, \mathcal{C}_2,\ldots, \mathcal{C}_{p+1}\}$ by picking $(p-1)$ operators from each class. Without loss of generality, we can index these operators as $\{\sigma_{i}^{(i)}, (\sigma_{i}^{(i)})^{2}, \ldots, (\sigma_{i}^{(i)})^{p-1}\}$, where $\sigma_{i}^{(i)}\in \cC_{i}$. Further, we parameterize a common eigenstate of this new class $\cC_{I}$ as follows:
\begin{align*}
|\psi\rangle \langle \psi | &=&  \frac{I}{p^2} + \sum_{l=1}^{p+1}\left[\omega ^ {e_l} \sigma_{l}^{l} + \omega ^ {2e_l} (\sigma_{l}^{l} )^2\right]  \\ 
&+& \ldots + \frac{1}{p^{2}}\omega ^ {(p-1)e_l} (\sigma_{l}^{l} )^{p-1} .
\end{align*}
The $H_2$ entropy of any state is given by:
\begin{align*}
H_2(\mathcal{B}_i||\psi\rangle) = -\log \sum_{\bf x =  a_1,a_2,\ldots,a_{p^2}}(Tr[|b_i^{\bf x}\rangle \langle b_i^{\bf x}|\psi\rangle \langle \psi |])^2
\end{align*}

\begin{align*}
= -\log \sum_{\bf x=a_1,a_2,\ldots,a_{p^2}}(\frac{1}{p^4} Tr\{ I + I \sum_{l=1}^{p+1} \omega^{x_{l}+(p-1)e_{l}} \\
   + \omega^{2x_{l}+(p-2)e_{l}} +\ldots +
   \omega^{(p-1)x_{l}+e_{l}}\})^2
  \end{align*}

For a given $\bf e$ and a given $l$, there are only $p$ vectors ${\bf x}$ with $x_{l}=e_{l}$. Therefore we have:
\begin{align*}
H_2(\mathcal{B}_i||\psi\rangle) &=  -\log~p \times \frac{1}{p^2} \\
   &= \log~p
\end{align*}
We see that this is true for any basis $\cB_{i} ~\forall~ i \in [1,2,\ldots,p+1]$ corresponding to any of the $p+1$ classes which were used to construct the new class $\cC_{I}$. Therefore we have,
\begin{equation*}
\frac{1}{p+1} \sum_{i=1}^{p+1}H_2(\mathcal{B}_i||\psi \rangle) = \log_2~p
\end{equation*}
This implies that the $H_{2}$-entropic uncertainty relation in Eq.~\eqref{eq:H2p} is tight for these $p+1$ measurement bases.

We now show that the Shannon entropic uncertainty relation in Eq.~\eqref{eq:H1p} is also tight for any prime squared dimensions. The  Shannon entropy associated with a measurement of $\cB_{i}$ on state $|\psi\rangle$ is given by,
\[
H(\mathcal{B}_i||\psi\rangle) = -\sum_{x=a_1,a_2,\ldots,a_{p^2}}  p_{i, |\psi\rangle}^{x}  \log p_{i, |\psi\rangle}^{x}  , \]
where, the probability $p_{i, |\psi\rangle}^{x}$ of obtaining outcome ${\bf x}$ when measuring basis $\cB_{i}$ on state $|\psi\rangle$ is given by
\begin{eqnarray}
p_{i, |\psi\rangle}^{\bf x} &=& \tr[\,|b_i^{x}\rangle \langle b_i^{x}|\psi\rangle \langle \psi| \,] \nonumber \\
&=& -\frac{1}{p^2}  + \frac{1}{p^2}\tr \sum_{l=1}^{p+1} \left[ \omega^{x_{l}+(p-1)e_{l}} +
   \omega^{2x_{l}+(p-2)e_{l}}  \right] \nonumber \\
   &+& \ldots + \frac{1}{p^{2}} \sum_{l=1}^{p+1}    \omega^{(p-1)x_{l}+e_{l}} . \nonumber
\end{eqnarray}   
Recall that, for a given $l$, $x_l = e_l$ only for $p$ $\bf x \in \{a_{1},a_{2},\ldots,a_{p^{2}}\}$. Hence, we have:
  \begin{align*}
   H(\mathcal{B}_j||\psi\rangle) = -p\times (\frac{p}{p^2}\log~\frac{1}{p}) = \log~p
\end{align*}
Therefore we have:
\begin{equation*}
\frac{1}{p+1} \sum_{i=1}^{p+1}H(\mathcal{B}_i||\psi \rangle) = \log~p ,
\end{equation*}
thus proving that the Shannon uncertainty relation in Eq.~\eqref{eq:H1p}is tight for these $p+1$ measurement bases.
  \end{proof}


\end{document}